\documentclass[11pt,letterpaper]{article}
\usepackage[T1]{fontenc}
\usepackage{lmodern,microtype}
\usepackage{needspace,float}
\usepackage[margin=1in]{geometry}
\usepackage{amsmath,amssymb,amsthm,mathtools,booktabs,enumitem}
\usepackage{tikz}
\usetikzlibrary{arrows.meta,positioning,fit,calc}
\usepackage[colorlinks=true,linkcolor=blue!55!black,citecolor=blue!55!black,urlcolor=blue!55!black]{hyperref}
\usepackage{bookmark}
\hypersetup{pdftitle={Counting Linear Extensions Faster than 2 to the n},pdfauthor={Keigo Oka},pdfsubject={Exact counting of linear extensions},pdfkeywords={linear extensions, exact exponential algorithms, posets}}
\newtheorem{theorem}{Theorem}
\newtheorem{lemma}[theorem]{Lemma}

\theoremstyle{definition}

\theoremstyle{remark}

\newcommand{\Ideals}{\mathcal I}
\newcommand{\LE}{\mathcal L}

\newcommand{\Ostar}{O^*}

\title{Counting Linear Extensions Faster than $2^n$}
\author{Keigo Oka\thanks{Google. Work done in a personal capacity.}\\\texttt{ogiekako@gmail.com}}
\date{October 26, 2026}

\begin{document}
\let\footnotesize\normalsize
\hypersetup{pageanchor=false}
\maketitle
\begin{abstract}
\normalsize
We give a deterministic algorithm that counts the linear extensions of
any $n$-element poset in $O^*(33^{n/6}) \le O^*(1.791^n)$ time and space,
breaking the $O^*(2^n)$ barrier of standard order-ideal dynamic
programming without assumptions on height, width, or dimension.

Our algorithm fixes an antichain $B$ and tracks order ideals only on the
remaining core $V \setminus B$. For each element $v \in B$ that has
successors in the core, we guess which of those successors appears first
in the linear extension, reducing all upper constraints on $v$ to a
single deadline. An element $v \in B$ becomes available once all of its
core predecessors have been placed; available elements with the same
deadline then face identical remaining constraints, so the dynamic program
only needs to record how many available elements have each deadline and
weight transitions by the number of choices for the next element.

Taking $B$ to be a maximum antichain and balancing this dynamic program
against a Dilworth chain partition yields an $O^*(1.89^n)$
algorithm. To obtain our main bound, we use a minimum-cut computation in
an auxiliary bipartite network to find an antichain $B$ whose comparability
neighborhood $R$ satisfies $|B| \ge 5|R|$ while the remainder $V \setminus (B \cup R)$
partitions into stars with at most five leaves. Since a $j$-leaf star has
$2^j+1$ independent sets and $(2^j+1)^{1/(j+1)} \le 33^{1/6}$ for
$1 \le j \le 5$, combining the star bound on the remainder with our
state bound on $B \cup R$ gives $O^*(33^{n/6})$ time and space.
\end{abstract}
\thispagestyle{empty}
\newpage
\setcounter{page}{1}
\hypersetup{pageanchor=true}

\section{Introduction}
A \emph{linear extension} of a finite poset $P=(V,<)$ is a total ordering
of $V$ that respects $<$, or equivalently a topological ordering of the
corresponding directed acyclic graph. Computing the number $e(P)$ of
linear extensions of an $n$-element poset is a basic problem in order
theory and exact algorithms~\cite{CP25,FK10}.

The standard dynamic program over order ideals (downward-closed subsets of
$P$) computes $e(P)$ in $O(n2^n)$ time and $O^*(2^n)$ space~\cite{DMB06}. In the
extreme case where $P$ is an antichain, all $2^n$ subsets of $V$ are order
ideals even though all elements are symmetric and $e(P) = n!$. This
suggests compressing dynamic-programming states across elements that face
identical remaining ordering constraints.

Computing $e(P)$ is \#P-complete~\cite{BW91}, even for posets of height
two~\cite{DP20} or dimension two~\cite{DP18}. On the algorithmic side,
Koivisto posed in 2013 the open problem of whether $e(P)$ can be computed
in $O((2-\varepsilon)^n)$ time for arbitrary posets~\cite[Section~6.13,
p.~67]{HPSW13}. For two-dimensional posets, Kozma~\cite{Kozma20} achieved
$O(1.8286^n)$ time (improved to $O(1.8172^n)$ in an updated preprint) by
fixing an antichain $B$, grouping elements of $B$ with identical
comparability neighborhoods in the core $V \setminus B$, and using
two-dimensionality to bound the number of groups.
For general posets, however, elements of $B$ can have $2^{|V\setminus B|}$
distinct neighborhoods in $V \setminus B$, preventing direct grouping.

In this paper, we answer Koivisto's question affirmatively for arbitrary
posets.

\begin{theorem}\label{thm:main}
Given an $n$-element poset $P$, the number of its linear extensions
$e(P)$ can be computed deterministically in $O^*(33^{n/6}) \le O^*(1.791^n)$
time and space.
\end{theorem}

\subsection{Overview of techniques}

\paragraph{Deadline compression on an antichain.}
Fix an antichain $B \subseteq V$ and refer to $X := V \setminus B$ as the
\emph{core}. In any linear extension, an element $v \in B$ becomes
available to be placed as soon as all of its core predecessors $D_v$ have
appeared; on the upper side, $v$ must precede all of its core successors
$S_v$. Even when two elements $u, v \in B$ have different upper sets
$S_u \neq S_v$, we can unify their upper constraints by guessing, for each
$v \in B$ with $S_v \neq \varnothing$, which element $g(v) \in S_v$
appears \emph{first} among $S_v$ in the linear extension. Restricting to
linear extensions in which $g(v)$ precedes $S_v \setminus \{g(v)\}$
reduces the upper constraint on $v$ to a single \emph{deadline}: $v$
simply needs to be placed before $g(v)$. Once all of its core predecessors
have appeared, an element of $B$ is \emph{released} (available to be
placed), and all released elements with the same deadline $c$ are
interchangeable with respect to future constraints. The DP state therefore
records only the number $q_c$ of released unplaced elements for each
deadline $c$. Placing one such element decrements $q_c$ and multiplies the
transition weight by $q_c$, accounting for the choice of which released
element to place next.

\paragraph{Multiplicity uniqueness and an $O^*(1.89^n)$ bound.}
We call a choice map $g = (g(v))_{v \in B, S_v \neq \varnothing}$ a
\emph{pattern}. Although there are $\prod_v |S_v|$ patterns in total, we
show that any pattern realizable by a linear order on the upper core
$Y = \bigcup_{v \in B} S_v$ is uniquely determined by its multiplicity
vector $(M_y)_{y \in Y}$, where $M_y = |\{v \in B : g(v) = y\}|$
(Lemma~\ref{first-lem:profile}).
Summing the number of DP states over all compatible patterns therefore
reduces to a sum over nonnegative integer tuples summing to $|B|$, which
evaluates to a single binomial coefficient. When the maximum antichain $B$ has size $n-k$,
Dilworth's chain cover bounds the number of order ideals by
$2^n(3/4)^k$, while this dynamic program summed over all patterns takes
$O^*(2^k\binom{n}{k})$ time. Balancing the two at
$k \approx 0.198n$ yields $O^*(1.89^n)$ time
(Section~\ref{sec:first}, Theorem~\ref{thm:first}).

\paragraph{Flow-based star decomposition and the $O^*(1.791^n)$ bound.}
In Section~\ref{sec:first}, the core $X = V \setminus B$ of a maximum
antichain is covered entirely by the neighborhood $R = N(B)$ in the
comparability graph, and we bound order ideals on $X$ by
$2^{|X|}$. To improve the exponential base, we relax $B$ from a maximum antichain to
an antichain with a small comparability neighborhood $R = N(B)$, leaving
a remainder $C = V \setminus (B \cup R)$ whose elements are incomparable
with every element of $B$. Using a minimum cut in a bipartite network with
source capacities $1$ and sink capacities $5$ (Lemma~\ref{lem:cut}), we
find an antichain $B$ with $|B| \ge 5|R|$ such that the saturated flow
between the left and right copies of $C$ defines a loopless functional
digraph on $C$ of maximum indegree $5$. We then
partition $C$ into stars with $1$ to $5$ leaves (Lemma~\ref{lem:stars}). Since a $j$-leaf star has $2^j+1$
independent sets and $(2^j+1)^{1/(j+1)} \le 33^{1/6}$ for $1 \le j \le 5$,
the restriction of core ideals to $C$ has at most $33^{|C|/6}$ choices.
In addition, the condition $|B| \ge 5|R|$ ensures that the binomial state
bound on $B \cup R$ is at most $10^{|R|}(6/5)^{|B|+|R|} \le 33^{(|B|+|R|)/6}$.
Combining the two factors bounds the total number of DP states across all
patterns by $O^*(33^{n/6}) \le O^*(1.791^n)$
(Sections~\ref{sec:decomp} and~\ref{sec:bound}).
Both the deadline-DP framework of Section~\ref{sec:first} and the
flow-star decomposition of Sections~\ref{sec:decomp}--\ref{sec:bound}
were developed through AI-assisted mathematical exploration with Claude
Opus~5 and ChatGPT, with all proofs and references verified by the author;
see the AI methodology statement at the end of the paper for details.

\section{Preliminaries}
Let $P=(V,<_P)$ be a finite poset on $n = |V|$ elements. We write
$u \le_P v$ when $u=v$ or $u <_P v$, and for $X \subseteq V$ we denote by
$P[X]$ the induced subposet on $X$. We assume $P$ is given by its comparability relation or by an acyclic
directed graph whose transitive closure defines $<_P$. The transitive
closure of $P$ and of any acyclic augmentation can be computed in
polynomial time. Throughout, $O^*(f(n))$ suppresses factors polynomial in
$n$. Since $e(P) \le n!$, all counts in our dynamic programs fit in
$O(n \log n)$ bits, so arithmetic operations on counts take polynomial
time.

An \emph{order ideal} (or simply \emph{ideal}) of $P$ is a subset
$I \subseteq V$ that is downward-closed ($u <_P v \in I \implies u \in I$).
We write $\Ideals(P)$ for the family of ideals of $P$ and $\LE(P)$ for
the set of linear extensions of $P$, so $e(P) = |\LE(P)|$. The map taking
an ideal $I$ to its set of maximal elements $\max_P(I)$ is a bijection
between $\Ideals(P)$ and the family of antichains of $P$.

The \emph{comparability graph} $G(P)$ is the undirected graph on $V$ with
an edge between each comparable pair, and its independent sets are
precisely the antichains of $P$. For a vertex subset $S$ in an undirected graph $G$,
we write $N_G(S) := \bigcup_{v \in S} N_G(v)$ for the union of the
neighborhoods of vertices in $S$.

\section{The deadline dynamic program and an \texorpdfstring{$O^*(1.89^n)$}{O*(1.89\textasciicircum n)} algorithm}\label{sec:first}
In this section we develop the first-upper-element patterns, the
multiplicity uniqueness lemma, and the deadline dynamic program, and combine
them with a chain partition to establish our first bound below $2^n$.

\begin{theorem}\label{thm:first}
Given an $n$-element poset $P$, the number of its linear extensions $e(P)$
can be computed deterministically in $O^*(1.89^n)$ time and space.
\end{theorem}

\subsection{Bounding order ideals via a chain partition}

In the standard order-ideal dynamic program~\cite{DMB06}, the last element
of any linear extension of $P[I]$ for a nonempty ideal $I \in \Ideals(P)$
belongs to $\max_P(I)$. Thus the recurrence
\[
F(\varnothing)=1,
\qquad
F(I)=\sum_{v\in\max_P(I)}F(I\setminus\{v\})
\]
computes $F(V)=e(P)$ in $\Ostar(|\Ideals(P)|)$ time and space by
generating $\Ideals(P)$ layer by layer in increasing order of cardinality.

\begin{lemma}[Dilworth's chain--antichain decomposition~\cite{Dil50}]\label{first-lem:chains}
In polynomial time, one can find a maximum antichain $B$ of $P$ and a
partition of $V$ into $|B|$ chains.
\end{lemma}
\begin{proof}
By Fulkerson's reduction to bipartite matching~\cite{Ful56}, form the
bipartite graph on left and right copies $V_L, V_R$ of $V$ with an edge $u_L v_R$ whenever
$u <_P v$, and compute a maximum matching $M$ and a minimum vertex cover
$K$ of size $m = |M|$ in polynomial time via K\"onig's theorem. Since $P$
is transitive and acyclic, the arcs $u \to v$ corresponding to $M$ have
indegree and outdegree at most one and therefore partition $V$ into $n-m$
chains. The set $B := \{v \in V : v_L \notin K \text{ and } v_R \notin K\}$
has size at least $n-m$ and contains no comparable pair since $K$ is a
vertex cover. Since no antichain can intersect a chain more than once,
$B$ is a maximum antichain of size $n-m$.
\end{proof}

Let $b := \operatorname{width}(P) = |B|$ be the maximum size of an
antichain and set $k := n-b$. Let $\ell_1,\ldots,\ell_b$ be the sizes
of the $b$ chains from Lemma~\ref{first-lem:chains}. Since every
antichain of $P$ contains at most one element from each chain, the
bijection between ideals and antichains gives
\[
|\Ideals(P)|\le \prod_{i=1}^b(\ell_i+1).
\]
For every integer $\ell\ge1$, induction shows
\begin{equation}\label{first-eq:chain-one}
\ell+1\le 2^\ell\left(\frac34\right)^{\ell-1},
\end{equation}
with equality for $\ell \in \{1,2\}$ and using $\frac{\ell+2}{\ell+1} \le \frac{3}{2}$ for the inductive step.
Multiplying~\eqref{first-eq:chain-one} over all $b$ chains and using $\sum_{i=1}^b \ell_i = n$ and $\sum_{i=1}^b (\ell_i-1) = n-b = k$ gives
\begin{equation}\label{first-eq:ideal-bound}
|\Ideals(P)|\le 2^n\left(\frac34\right)^k.
\end{equation}

\subsection{First-upper-element patterns}

Fix a maximum antichain $B\subseteq V$ and let $X:=V\setminus B$ be the
core, of size $|X|=k=n-b$. For each $v\in B$, define its core
predecessors and successors by
\[
D_v:=\{x\in X:x<_P v\},
\qquad
S_v:=\{x\in X:v<_P x\},
\]
and set
\[
\Lambda:=\bigcup_{v\in B}D_v,
\qquad
Y:=\bigcup_{v\in B}S_v.
\]
Since $B$ is a maximum (and hence maximal) antichain, every $x\in X$ is
comparable to some element of $B$, so $X=\Lambda\cup Y$. Moreover,
$\Lambda$ and $Y$ are disjoint, since if $x\in \Lambda\cap Y$, then
$v_1<_P x<_P v_2$ for some $v_1,v_2\in B$, contradicting that $B$ is an
antichain. Since replacing $P$ with its dual poset $(V, >_P)$ preserves
$e(P)$ and exchanges $\Lambda$ with $Y$, we may assume without loss of
generality that
\begin{equation}\label{first-eq:s-bound}
s:=|Y|\le k/2.
\end{equation}

For each $v\in B$ with $S_v\ne\varnothing$, we select an element
$g(v)\in S_v$ and impose the ordering constraints
\begin{equation}\label{first-eq:pattern-relations}
g(v)<y
\qquad
\text{for all } y\in S_v\setminus\{g(v)\}.
\end{equation}
We call such a map $g$ a \emph{pattern}. A pattern $g$ is
\emph{compatible} if augmenting $P$ with the relations~\eqref{first-eq:pattern-relations}
over all $v \in B$ with $S_v \neq \varnothing$ yields an acyclic relation,
in which case we denote its transitive closure by $P_g$.

Every linear extension $\pi \in \LE(P)$ induces a unique compatible
pattern $g$ by setting $g(v)$ to be the earliest element of $S_v$ in
$\pi$. Conversely, for any compatible pattern $g$, a permutation lies in
$\LE(P_g)$ if and only if it belongs to $\LE(P)$ and places $g(v)$ first
among $S_v$ for every $v \in B$ with $S_v \neq \varnothing$. Thus $\LE(P)$
partitions as
\begin{equation}\label{first-eq:partition}
\LE(P)=\mathop{\dot\bigcup}_{g\text{ compatible}}\LE(P_g).
\end{equation}
A compatible pattern is uniquely determined by how many elements of $B$
select each element of $Y$ as their first successor.

\begin{lemma}[Multiplicity uniqueness]\label{first-lem:profile}
Let $S_1,\ldots,S_t$ be nonempty subsets of $Y$. Call a choice tuple
$(g_1,\ldots,g_t) \in S_1 \times \cdots \times S_t$ \emph{realizable} if
there is a total order on $Y$ in which $g_i$ is the minimum element of
$S_i$ for every $i \in \{1,\ldots,t\}$. Then any realizable choice tuple
is uniquely determined by its multiplicity vector
$(m_y)_{y \in Y}$, where $m_y := |\{i : g_i = y\}|$.
\end{lemma}
\begin{proof}
Fix a total order on $Y$ realizing $g=(g_1,\ldots,g_t)$, and identify $Y$
with $\{1,\ldots,|Y|\}$ according to this order. For any other choice
tuple $g'=(g'_1,\ldots,g'_t)$ with the same multiplicities $(m_y)_{y\in Y}$,
we have $\sum_{i=1}^t g'_i = \sum_{y\in Y} y\,m_y = \sum_{i=1}^t g_i$.
On the other hand, since $g_i = \min S_i$ in this order and $g'_i \in S_i$,
we have $g'_i \ge g_i$ for every $i$. Equality of the sums therefore
forces $g'_i = g_i$ for all $i$.
\end{proof}

Applying Lemma~\ref{first-lem:profile} to the family $(S_v)_{v\in B,\, S_v\ne\varnothing}$
shows that every compatible pattern $g$ (being realized by the
restriction of any linear extension of $P_g$ to $Y$) is uniquely
determined by the multiplicity vector
\[
 M_y := |\{v\in B:g(v)=y\}|\qquad(y\in Y).
\]

\paragraph{Enumerating compatible patterns.}
We enumerate all compatible patterns by simple backtracking over the
$t \le b$ elements $v \in B$ with $S_v \neq \varnothing$ in a fixed order,
branching on $g(v) \in S_v$ and pruning any partial assignment whose
added relations~\eqref{first-eq:pattern-relations} create a directed cycle.
Every acyclic partial assignment extends to at least one compatible full
pattern, since fixing any topological ordering of the current acyclic
poset and setting $g(u)$ to be the earliest element of $S_u$ for each
remaining $u$ yields a valid completion. Since distinct partial
assignments at the same level have disjoint sets of completions, there are
at most $N_0$ active partial assignments at each of the $t$ levels of the
search tree, where $N_0 \ge 1$ is the total number of compatible full
patterns, and each node tests at most $s = |Y|$ choices. The backtracking search therefore
performs at most $s t N_0$ cycle checks and takes $\Ostar(N_0)$ time in
total.

\subsection{The deadline dynamic program}

Fix a compatible pattern $g$. For each $v\in B$, define its \emph{deadline}
\[
\delta(v):=
\begin{cases}
g(v),&\text{if } S_v\ne\varnothing,\\
\infty,&\text{if } S_v=\varnothing.
\end{cases}
\]
Since no element of $Y$ can precede any element of $B$ in $P$ (as
$u <_P y \le_P v$ with $u,v \in B$ would contradict that $B$ is an
antichain) and all added relations~\eqref{first-eq:pattern-relations} lie
entirely within $Y$, the set of core predecessors of $v \in B$ in $P_g$
is still $D_v$. On the upper side, if $\delta(v) = g(v) \in Y$, then
$v <_P g(v)$ together with~\eqref{first-eq:pattern-relations} implies
$v <_{P_g} g(v) <_{P_g} y$ for all $y \in S_v \setminus \{g(v)\}$. Hence,
once all elements of $D_v$ have been placed, the only remaining constraint
on $v$ in $P_g$ is that it must appear before $\delta(v)$.

Let $Q_g:=P_g[X]$ be the subposet of $P_g$ induced by the core $X$. For
an order ideal $I \in \Ideals(Q_g)$ and a deadline $c\in Y\cup\{\infty\}$,
let
\[
A_c(I):=|\{v\in B:\delta(v)=c\text{ and } D_v\subseteq I\}|
\]
be the number of antichain elements with deadline $c$ whose release
condition $D_v \subseteq I$ is satisfied at $I$ (we say such elements are
\emph{released}, or available to be placed).
A state of the dynamic program is a pair $(I,q)$, where $I \in \Ideals(Q_g)$
and $q=(q_c)_{c\in Y\cup\{\infty\}}$ records the number $q_c$ of released
but not yet placed elements of $B$ with deadline $c$. The initial state
is $(\varnothing, (A_c(\varnothing))_{c\in Y\cup\{\infty\}})$, the
terminal state is $(X,0)$, and the transitions from $(I,q)$ are:
\begin{enumerate}[label=(\roman*),leftmargin=2.2em]
\item \emph{(Place an antichain element)} For any $c \in Y \cup \{\infty\}$
with $q_c>0$, transition to $(I, q')$ with weight $q_c$, where $q'$ is
obtained from $q$ by decrementing $q_c$ by $1$.
\item \emph{(Place a core element)} For any $x\in X\setminus I$ such that
$I':=I\cup\{x\} \in \Ideals(Q_g)$ and such that $q_x=0$ if $x\in Y$,
transition with weight $1$ to $(I', q')$, where
\[
q'_c = q_c+A_c(I')-A_c(I) \qquad \text{for all } c \in Y \cup \{\infty\}.
\]
\end{enumerate}

\begin{lemma}[Correctness of the deadline DP]\label{first-lem:deadline}
For any compatible pattern $g$, the total weight of all paths from the
initial state to $(X,0)$ equals $e(P_g)$.
\end{lemma}

\begin{proof}
Any prefix of a linear extension of $P_g$ determines the set $I \in \Ideals(Q_g)$
of core elements placed so far and, for each $c \in Y \cup \{\infty\}$,
the number $q_c$ of released elements $v \in B$ (i.e., $D_v \subseteq I$)
with $\delta(v)=c$ that have not yet been placed. Whenever a core element
$x \in Y$ is eligible to be placed next ($I \cup \{x\} \in \Ideals(Q_g)$),
every $v \in B$ with $\delta(v)=x$ has already been released, since any
$d \in D_v$ satisfies $d <_P v <_P x$ and $I \cup \{x\} \in \Ideals(Q_g)$
forces $d \in I$. Thus, requiring $q_x = 0$ before placing $x \in Y$
is necessary and sufficient to ensure that all $v \in B$ with deadline $x$
precede $x$, and placing $x$ never releases any new element with deadline $x$.

Since all released unplaced elements of $B$ with the same deadline $c$
are incomparable with one another and share the sole remaining constraint
of preceding $c$, any of the $q_c$ such elements can be placed next in a
transition of type~(i), and all $q_c$ choices lead to the same state
$(I, q')$. By induction on the length of a DP path, the product of
transition weights along the path equals the number of valid prefixes of
linear extensions of $P_g$ that follow that sequence of states. Summing
over all paths ending at $(X,0)$ therefore yields $|\LE(P_g)| = e(P_g)$.
\end{proof}

\begin{figure}[!ht]
\centering
\begin{tikzpicture}[>=Latex,thick,font=\normalsize,
 core/.style={draw,rectangle,minimum width=7mm,minimum height=6mm,fill=white},
 lab/.style={draw,circle,minimum size=7mm,fill=white}]
\node[core] (l) at (0,0) {$\lambda$};
\node[core] (m) at (3.2,0) {$\mu$};
\node[lab] (b) at (-0.9,1.3) {$b$};
\node[lab] (a) at (0.9,1.3) {$a$};
\node[lab] (d) at (3.2,1.3) {$d$};
\node[core] (y) at (0,2.7) {$y$};
\node[core] (z) at (3.2,2.7) {$z$};
\draw[->] (l)--(a); \draw[->] (l)--(b); \draw[->] (m)--(d);
\draw[->] (a)--(y); \draw[->] (a)--(z); \draw[->] (b)--(y); \draw[->] (d)--(z);
\draw[->,dashed,line width=1pt] (y) to[bend left=15] node[above] {$g(a)=y$} (z);
\node[text width=6.2cm,align=left,anchor=west] at (4.3,1.45)
 {$B=\{a,b,d\}$, $\Lambda=\{\lambda,\mu\}$, $Y=\{y,z\}$.\\
 $g(a)=g(b)=y$, $g(d)=z$.\\[4pt]
 At $I=\{\lambda\}$:\\
 $q_y=2$, $q_z=0$; $d$ is not released.\\[4pt]
 Placing $a$ or $b$ leads to the same state with $q_y=1$ and transition weight $2$.};
\end{tikzpicture}
\caption{An example of deadline compression on $B=\{a,b,d\}$ (circles) with
core elements $\Lambda=\{\lambda,\mu\}$ and $Y=\{y,z\}$ (boxes). Although
$S_a=\{y,z\}$ and $S_b=\{y\}$ differ, choosing $g(a)=y$ adds the dashed arc
$y \to z$, giving $a$ and $b$ the same deadline $y$. Once $\lambda$ is
placed, $a$ and $b$ are counted together in $q_y=2$.}
\label{fig:queue}
\end{figure}

Figure~\ref{fig:queue} illustrates a transition of the deadline DP.
By~\eqref{first-eq:partition} and Lemma~\ref{first-lem:deadline}, summing
the output of the deadline DP over all compatible patterns $g$ yields $e(P)$.

\subsection{State count and proof of Theorem~\ref{thm:first}}

Let $M_\infty := |\{v\in B:S_v=\varnothing\}|$, so that
$M_\infty + \sum_{u\in Y} M_u = b$.
For a single compatible pattern $g$, the number of states $(I,q)$ of the
deadline DP is at most
\begin{equation}\label{first-eq:state-one}
2^{k-s}(M_\infty+1)\prod_{u\in Y}(M_u+2).
\end{equation}
Indeed, $I \cap \Lambda$ has at most $2^{|\Lambda|} = 2^{k-s}$ choices,
$q_\infty \in \{0,\ldots,M_\infty\}$ has $M_\infty+1$ choices, and for
each $u \in Y$, either $u \notin I$ and $q_u \in \{0,\ldots,M_u\}$
($M_u+1$ choices) or $u \in I$ and $q_u = 0$ ($1$ choice), giving $M_u+2$
joint possibilities for $(\mathbf{1}_{u\in I}, q_u)$.

By Lemma~\ref{first-lem:profile}, distinct compatible patterns have
distinct multiplicity vectors. Writing $Y=\{u_1,\ldots,u_s\}$,
$m_i=M_{u_i}$, and $m_0=M_\infty$, we can upper-bound the total number of
states across all compatible patterns by summing~\eqref{first-eq:state-one}
over all nonnegative integer tuples $(m_0,\ldots,m_s)$ summing to $b$:
\begin{align}
\#\{\text{states over all patterns}\}
&\le 2^{k-s}\sum_{m_0+\cdots+m_s=b}(m_0+1)\prod_{i=1}^s(m_i+2)\nonumber\\
&\le 2^k\sum_{m_0+\cdots+m_s=b}\prod_{i=0}^s(m_i+1)\nonumber\\
&=2^k\binom{b+2s+1}{2s+1},
\label{first-eq:state-sum}
\end{align}
where the second line uses $m_i+2 \le 2(m_i+1)$ for $1 \le i \le s$, and
the equality holds because there are $m_i+1$ ways to write each $m_i$ as a
sum $m'_i + m''_i$ of two nonnegative integers, so the sum counts the
$\binom{b+2s+1}{2s+1}$ nonnegative integer tuples
$(m'_0, m''_0, \ldots, m'_s, m''_s)$ that sum to $b$.

When $k < n/2$, inequality~\eqref{first-eq:s-bound} gives
$2s+1 \le k+1 \le \frac{n+1}{2}$ and $b+2s+1 \le n+1$, so monotonicity
of binomial coefficients yields
\begin{equation}\label{first-eq:binom-bound}
\binom{b+2s+1}{2s+1}
\le \binom{n+1}{2s+1}
\le \binom{n+1}{k+1}
=\frac{n+1}{k+1}\binom{n}{k}.
\end{equation}
Evaluating the DP layer by layer by the total number of placed elements
$0,1,\ldots,n$ therefore takes $\Ostar(2^k\binom{n}{k})$
time and space in total.

\begin{proof}[Proof of Theorem~\ref{thm:first}]
Compute a maximum antichain $B$ and chain partition by
Lemma~\ref{first-lem:chains}, let $k = n - |B|$, and set the threshold
$\alpha := 0.198$.
If $k \ge \alpha n$, we run the standard order-ideal DP. By~\eqref{first-eq:ideal-bound},
its running time is at most
\[
\Ostar\!\left(2^n(3/4)^k\right)
\le \Ostar\!\left(\left[2(3/4)^\alpha\right]^n\right)
= \Ostar(1.89^n),
\]
since $2(3/4)^{0.198} < 1.8899 < 1.89$.

If $k < \alpha n$, we run the pattern-enumeration and deadline DP
algorithm, which by~\eqref{first-eq:state-sum} and~\eqref{first-eq:binom-bound}
takes $\Ostar(2^k\binom{n}{k})$ time and space. Setting
$x := \frac{\alpha}{2(1-\alpha)} = \frac{99}{802} \in (0,1)$, the
binomial theorem gives $(1+2x)^n \ge 2^k \binom{n}{k} x^k$, and since
$k < \alpha n$,
\[
2^k\binom{n}{k}
\le \left[(1+2x)x^{-\alpha}\right]^n
< 1.89^n,
\]
where $(1+2x)x^{-\alpha} = \frac{500}{401}\left(\frac{802}{99}\right)^{0.198} < 1.8899 < 1.89$.
Thus both cases run in $O^*(1.89^n)$ time and space.
\end{proof}

\section{Flow decomposition into an antichain and small stars}\label{sec:decomp}
The pattern partition~\eqref{first-eq:partition}, Lemma~\ref{first-lem:profile},
and the deadline DP (Lemma~\ref{first-lem:deadline}) apply to any
antichain $B$, not only to a maximum antichain. When $B$ is not maximal,
the core $X = V \setminus B$ splits into the comparability neighborhood
$R = N_{G(P)}(B)$ and a remainder $C = V \setminus (B \cup R)$ whose
elements are incomparable with every element of $B$. We now show how a
single minimum-cut computation (illustrated in Figure~\ref{fig:flow})
finds an antichain $B$ with $|B| \ge 5|R|$ such that $C$ can be
partitioned into stars with at most five leaves.

\begin{lemma}\label{lem:cut}
In polynomial time, the vertex set of any undirected simple graph $G=(V,E)$
can be partitioned into $B,R,C$ such that:
\begin{enumerate}[label=(\roman*),itemsep=2pt]
\item $B$ is an independent set, $R=N_G(B)$, and $|B|\ge5|R|$;
\item there is a map $f:C\to C$ such that $vf(v) \in E$ and $f(v) \neq v$
for all $v \in C$, and $|f^{-1}(u)| \le 5$ for all $u \in C$.
\end{enumerate}
\end{lemma}
\begin{proof}
Construct a flow network on left and right copies $V_L, V_R$ of $V$ with
a source $s$ and sink $t$: add an arc $s \to v_L$ of capacity $1$ for each
$v \in V$, an arc $v_R \to t$ of capacity $5$ for each $v \in V$, and for
each undirected edge $uv \in E$, both directed arcs $u_L \to v_R$ and
$v_L \to u_R$ of capacity $|V|+1$. Compute a maximum integral $s$--$t$
flow in polynomial time, and let $(K_s, K_t)$ be the minimum $s$--$t$ cut
where $K_s$ is the set of vertices reachable from $s$ in the residual
network.

Since the cut $(\{s\}, V_L \cup V_R \cup \{t\})$ has capacity $|V|$,
no adjacency arc of capacity $|V|+1$ can cross from $K_s$ to $K_t$. Let
$S := \{v \in V : v_L \in K_s\}$. Every neighbor of $S$ in $V_R$ must lie
in $K_s$ to avoid cutting an adjacency arc, and including any other
vertex of $V_R$ would strictly increase the cut capacity, so
$\{v \in V : v_R \in K_s\} = N_G(S)$. The capacity of $(K_s, K_t)$ is
therefore $|V| - |S| + 5|N_G(S)|$, and $S$ maximizes $|S| - 5|N_G(S)|$.
In particular, comparing $S$ with $\varnothing$ gives $|S| \ge 5|N_G(S)|$.

Moreover, $S$ is an independent set. Indeed, if we set $U := S \cap N_G(S)$ and
$B' := S \setminus N_G(S)$, then $B'$ is independent, $N_G(B') \cap S = \varnothing$,
and $U$ and $N_G(B')$ are disjoint subsets of $N_G(S)$, which implies
\[
 |B'|-5|N_G(B')|\ge (|S|-|U|)-5(|N_G(S)|-|U|) = |S|-5|N_G(S)|+4|U|.
\]
Since $S$ maximizes $|S|-5|N_G(S)|$ over all subsets of $V$, this
inequality forces $U = \varnothing$. Thus $B := S$ is an independent set
and $R := N_G(B)$ satisfies $|B| \ge 5|R|$.

Now let $C := V \setminus (B \cup R)$. For every $v \in C$, we have
$v_L \notin K_s$, so the source arc $s \to v_L$ is saturated and sends
one unit of integral flow out of $v_L$ to some $u_R$. Here $u \notin B$ because $v \notin N_G(B) = R$, and $u \notin R$ because
otherwise $u_R \in K_s$ and positive flow on $v_L \to u_R$ would make
$v_L$ reachable from $u_R$ along the reverse residual arc, contradicting
$v_L \notin K_s$. Hence $u \in C$.
Setting $f(v) := u$ gives a map $f: C \to C$ along edges of $G$ with no
fixed points, and the sink capacities of $5$ ensure $|f^{-1}(u)| \le 5$
for every $u \in C$.
\end{proof}

\begin{figure}[!ht]
\centering
\begin{tikzpicture}[>=Latex,thick,font=\normalsize]
\node[circle,draw,minimum size=7mm] (s) at (0,0) {$s$};
\node[draw,rounded corners,minimum width=15mm,minimum height=14mm,align=center] (vl) at (2.1,0) {$V_L$\\left copy};
\node[draw,rounded corners,minimum width=15mm,minimum height=14mm,align=center] (vr) at (5.4,0) {$V_R$\\right copy};
\node[circle,draw,minimum size=7mm] (t) at (7.5,0) {$t$};
\draw[->] (s)--node[above] {$1$} (vl);
\draw[->] (vl)--node[above] {$uv\in E$} (vr);
\node at (3.75,-1.05) {capacity $n+1$};
\draw[->] (vr)--node[above] {$5$} (t);
\node[circle,draw,fill=white,minimum size=4mm] (o) at (9.6,0) {};
\foreach \a in {18,90,162,234,306}{
  \node[circle,draw,fill=white,minimum size=4mm] (v\a) at ($(o)+(\a:0.85)$) {};
  \draw (o)--(v\a);
}
\node[below=12mm of o,align=center] {$2^5+1=33$\\independent sets};
\end{tikzpicture}
\caption{Left: the bipartite flow network of Lemma~\ref{lem:cut}. Right: a
six-vertex block $C_i \subseteq C$ containing a five-leaf star has at most
$2^5+1=33$ independent sets in $G(P)$, and hence $|\Ideals(P[C_i])| \le 33$
(Lemma~\ref{lem:stars}).}
\label{fig:flow}
\end{figure}

\begin{lemma}\label{lem:stars}
Given the set $C$ and map $f: C \to C$ from Lemma~\ref{lem:cut}, one can
partition $C$ in polynomial time into blocks $C_1,\ldots,C_t$ such that
each block $C_i$ induces a subgraph of $G$ containing a spanning star with
$1$ to $5$ leaves (so $2 \le |C_i| \le 6$). Therefore, when $G=G(P)$,
\[
 J_C:=\prod_{i=1}^t |\Ideals(P[C_i])|\le33^{|C|/6}.
\]
\end{lemma}
\begin{proof}
Consider the functional digraph on $C$ with arcs $v \to f(v)$. Since $f$
has no fixed points and $|f^{-1}(u)| \le 5$ for all $u \in C$, each weakly
connected component consists of a directed cycle of length at least $2$
together with trees rooted on the cycle and directed toward their roots.
Viewing $f(v)$ as the parent of each vertex $v$ outside the cycles, every
vertex has at most $5$ children.

First, as long as some vertex $u$ outside the cycles has at least one
child, choose a deepest such vertex $u$ in its tree. By the choice of $u$,
all of its children are leaves, so removing $u$ together with its children
forms a star block with $1$ to $5$ leaves. Repeating this step removes all
vertices outside the cycles except leaves attached directly to cycle
vertices. Each cycle vertex $w$ now has at most $4$ attached leaves, since
its predecessor on the cycle is also in $f^{-1}(w)$ and $|f^{-1}(w)| \le 5$.

Next, we partition each cycle together with its attached leaves (see
Figure~\ref{fig:star-partition}):
\begin{itemize}[leftmargin=2em,itemsep=2pt]
\item If no vertex on the cycle has attached leaves, the cycle has length
at least $2$ and can be partitioned into adjacent pairs (1-leaf stars)
and, if its length is odd, one consecutive triple (a 2-leaf star).
\item Otherwise, mark every cycle vertex that has at least one attached
leaf as a center. Partition the unmarked cycle vertices into maximal
contiguous runs along the cycle (between consecutive centers, or from a
single center back to itself). If a run has length $1$, we attach its
single vertex as one additional leaf to the center immediately preceding
the run along the cycle. If a run has length at least $2$, we partition
the run itself into adjacent pairs and at most one triple. Each marked
center, together with its attached leaves (including at most one absorbed
cycle vertex), then forms a star block with $1$ to $5$ leaves.
\end{itemize}

Finally, a block $C_i$ of size $j+1$ ($1 \le j \le 5$) contains a spanning
star with $j$ leaves in $G(P)$. A $j$-leaf star has $2^j+1$ independent
sets, and any additional comparability edges inside $C_i$ can only reduce
the number of independent sets. Under the bijection between ideals and
antichains, $|\Ideals(P[C_i])| \le 2^j+1$. Since $(2^j+1)^{1/(j+1)} \le 33^{1/6}$ for all $j \in \{1,2,3,4,5\}$, we
obtain $|\Ideals(P[C_i])| \le 33^{|C_i|/6}$. Multiplying over $i=1,\ldots,t$
gives $J_C \le 33^{|C|/6}$.
\end{proof}

We remark that the existence of a spanning star factor in $G[C]$ with $1$
to $5$ leaves also follows from the star-factor theorem of Amahashi and
Kano~\cite{AK82}: for every $S \subseteq C$, each isolated vertex $v$ of
$G[C] - S$ must have $f(v) \in S$, and $|f^{-1}(u)| \le 5$ implies that
$G[C] - S$ has at most $5|S|$ isolated vertices.

\begin{figure}[!ht]
\centering
\begin{tikzpicture}[thick,font=\normalsize,
 v/.style={circle,draw,fill=white,minimum size=4mm,inner sep=0pt},
 center/.style={v,fill=blue!25},
 block/.style={draw=blue!55!black,fill=blue!6,rounded corners=3pt}]
\begin{scope}
\node at (1,2.4) {(a) Remove a tree star};
\draw[block] (-0.35,-0.3) rectangle (2.35,1.25);
\node[v] (p) at (1,1.8) {};
\node[center] (c) at (1,0.9) {};
\draw[dashed] (c)--(p);
\foreach \i/\x in {1/0,2/1,3/2}{
  \node[v] (a\i) at (\x,0) {}; \draw (c)--(a\i);
}
\node[align=center] at (1,-0.75) {center and its\\leaf children};
\end{scope}
\begin{scope}[xshift=5.15cm]
\node at (1,2.4) {(b) Absorb a single vertex};
\draw[block] (-0.6,-0.3) rectangle (1.3,1.25);
\node[center] (l) at (0,0) {};
\node[v] (a) at (-0.25,0.9) {};
\node[v] (u) at (1,0) {};
\node[center] (r) at (2,0) {};
\node[v] (b) at (2.25,0.9) {};
\draw (a)--(l)--(u); \draw[dashed] (u)--(r); \draw (r)--(b);
\node[align=center] at (1,-0.75) {one extra leaf for\\the left center};
\end{scope}
\begin{scope}[xshift=10.2cm]
\node at (1.2,2.4) {(c) Split a longer run};
\draw[block] (-0.25,-0.3) rectangle (0.85,0.3);
\draw[block] (0.95,-0.3) rectangle (2.65,0.3);
\foreach \i/\x in {1/0,2/0.6,3/1.2,4/1.8,5/2.4}{
  \node[v] (z\i) at (\x,0) {};
}
\node[center] (l) at (-0.65,0) {};
\node[center] (r) at (3.05,0) {};
\draw[dashed] (l)--(z1) (z2)--(z3) (z5)--(r);
\draw (z1)--(z2) (z3)--(z4)--(z5);
\node at (0.3,0.85) {pair}; \node at (1.8,0.85) {triple};
\node[align=center] at (1.2,-0.75) {five vertices with\\no attached leaves};
\end{scope}
\end{tikzpicture}
\caption{Partitioning the functional digraph on $C$ into star blocks in
Lemma~\ref{lem:stars}. Shaded boxes indicate selected star blocks:
(a) pruning a deepest internal tree vertex (filled) with its leaf
children; (b) absorbing a single unmarked cycle vertex into the preceding
marked cycle center (filled); (c) splitting a run of two or more unmarked
cycle vertices between marked centers (filled) into pairs and at most one
triple.}
\label{fig:star-partition}
\end{figure}

\section{Proof of Theorem~\ref{thm:main}}\label{sec:bound}
We now combine the decomposition of Lemmas~\ref{lem:cut} and~\ref{lem:stars}
with the deadline dynamic program of Section~\ref{sec:first}.
Apply Lemma~\ref{lem:cut} to the comparability graph $G(P)$ to obtain the
partition $V = B \mathbin{\dot\cup} R \mathbin{\dot\cup} C$ with
$b := |B|$, $r := |R|$, and $b \ge 5r$.
On the core $X := V \setminus B = R \mathbin{\dot\cup} C$, define
$D_v, S_v, \Lambda, Y$ with respect to the antichain $B$ as in
Section~\ref{sec:first}. Since $R = N_{G(P)}(B)$ and every vertex of $C$
is incomparable with every element of $B$, we have $D_v, S_v \subseteq R$
for all $v \in B$, and thus $R = \Lambda \mathbin{\dot\cup} Y$.
Replacing $P$ with its dual poset preserves $G(P)$ and exchanges $\Lambda$
with $Y$, so we may assume $s := |Y| \le r/2$.
We then enumerate all compatible patterns $g$ on $(S_v)_{v \in B,\, S_v \neq \varnothing}$
and run the deadline dynamic program of Section~\ref{sec:first} for each
pattern $g$.

\begin{proof}[Proof of Theorem~\ref{thm:main}]
Correctness follows immediately from~\eqref{first-eq:partition} and
Lemma~\ref{first-lem:deadline}, which require only that $B$ be an
antichain. To bound the total number of states across all compatible
patterns, observe that for any fixed pattern $g$, the core poset
$Q_g = P_g[X]$ retains all original comparability relations of $P[X]$.
Thus, for every core ideal $I \in \Ideals(Q_g)$ and every star
block $C_i$ of $C$ from Lemma~\ref{lem:stars}, the intersection
$I \cap C_i$ is an order ideal of $P[C_i]$. Thus $I \cap C$ has at most
$J_C = \prod_{i=1}^t |\Ideals(P[C_i])| \le 33^{|C|/6}$ choices, while
$I \cap \Lambda$ has at most $2^{|\Lambda|} = 2^{r-s}$ choices.

For each choice of $I \cap C$ and $I \cap \Lambda$, the same count over
$Y$ and the available-element counts $q$ as in~\eqref{first-eq:state-one}--\eqref{first-eq:state-sum}
bounds the total number of states across all compatible patterns by
\[
 \#\{\text{states over all patterns}\}
 \le 2^{r-s} J_C \sum_{m_0+\cdots+m_s=b}(m_0+1)\prod_{i=1}^s(m_i+2)
 \le 2^r J_C \binom{b+2s+1}{2s+1}.
\]
Since $2s \le r$ and $\binom{b+t}{t}$ is increasing in $t \ge 0$, we have
\[
 \binom{b+2s+1}{2s+1}
 \le \binom{b+r+1}{r+1}
 = \frac{b+r+1}{r+1}\binom{b+r}{r}
 \le (n+1)\binom{b+r}{r}.
\]
By the binomial theorem,
\[
 \left(\frac65\right)^{b+r}
 =\sum_{j=0}^{b+r}\binom{b+r}{j}5^{-j}
 \ge \binom{b+r}{r}5^{-r},
\]
and since $b \ge 5r$ implies $r \le (b+r)/6$, we obtain
\[
 2^r\binom{b+r}{r}
 \le 10^r\left(\frac65\right)^{b+r}
 \le\left[10\left(\frac65\right)^6\right]^{(b+r)/6}
 \le 33^{(b+r)/6},
\]
using $10(6/5)^6 = 93312/3125 = 29.85984 < 33$.
Combining this with $J_C \le 33^{|C|/6}$ from Lemma~\ref{lem:stars} and
$|C| + b + r = n$ bounds the total number of DP states over all patterns by
\[
 (n+1)\,33^{|C|/6}\,33^{(b+r)/6}
 =(n+1)\,33^{n/6}.
\]
Since $33^{1/6} < 1.79096 < 1.791$, the algorithm runs in $O^*(33^{n/6}) \le O^*(1.791^n)$
time and space.
\end{proof}

\section{Further context and discussion}\label{sec:scope}
\paragraph{Sampling, volume estimation, and sorting.}
For approximate counting and uniform sampling of linear extensions,
polynomial-time randomized algorithms are known via rapidly mixing Markov
chains~\cite{KK91,BD99}, geometric sampling methods~\cite{Mat91}, perfect
sampling~\cite{Hub06}, and volume approximation of the order
polytope~\cite{Sta86,DFK91}. Linear extension counts
also govern the complexity of sorting under partial information. Kahn and
Saks~\cite{KS84} showed that every non-chain poset contains a pair of
elements $u,v$ such that the fraction of linear extensions placing $u$
before $v$ lies between $3/11$ and $8/11$, allowing any poset $P$ to be
sorted in $O(\log e(P))$ comparisons.

\paragraph{Parameterized and structural algorithms for exact counting.}
On posets of width $w$, Dilworth's chain partition bounds the number of
order ideals by $(n/w+1)^w$. Parameterized by treewidth, computing $e(P)$
is fixed-parameter tractable with respect to the incomparability graph,
but $W[1]$-hard with respect to the cover graph~\cite{EGKO19}; see
also~\cite{KKS20} for tree-decomposition algorithms and~\cite{TKNK18,TKNK18b}
for practical approximation schemes using exact counting subroutines.

\paragraph{Beating $2^n$ in subset dynamic programs.}
Improving over $O^*(2^n)$ subset dynamic programs has been studied for a
variety of problems in exact algorithms~\cite{Woe03,FK10}, including
undirected Hamiltonicity~\cite{Bjo14}, bipartite TSP (conditionally on
$\omega=2$)~\cite{Ned20}, and subset problems via monotone local
search~\cite{FGLS16}. Unlike approaches based on algebraic cancellation or
randomization, our algorithm is purely combinatorial, combining deadline
compression on an antichain with a minimum-cut star decomposition of the
remainder.

\paragraph{Open questions.}
Our algorithm uses $O^*(33^{n/6})$ space to store dynamic-programming
states across layers. Whether $e(P)$ can be computed in $O((2-\varepsilon)^n)$
time and \emph{polynomial} space, and how far the exponential base
$33^{1/6} \approx 1.791$ can be lowered (for instance, by higher-order
star or component packings in $C$), remain natural open problems.

\section*{AI methodology and disclosure}
Section~\ref{sec:first} adapts the argument of the initial preprint~\cite{Oka26},
where Claude Opus 5 (Anthropic) developed the first-upper-element pattern,
multiplicity uniqueness, and deadline-DP ideas during AI-assisted
exploration, and ChatGPT 5.6 Sol (OpenAI) assisted with the research
prompt. OpenAI ChatGPT was also used to assist with developing the
flow-star decomposition and state bound in Sections~\ref{sec:decomp}--\ref{sec:bound}
and drafting the exposition and figures. The author verified the
correctness and originality of all content, including all proofs and
references, and takes full responsibility for the manuscript.

\end{document}